\documentclass{article}
\usepackage[a4paper, margin=1.5in]{geometry}
 
\usepackage{microtype}
\usepackage{algorithm}
\usepackage{algorithmicx}
\usepackage[noend]{algpseudocode}
\usepackage[moderate]{savetrees}
\usepackage{amsthm,amsmath,amssymb,graphicx,mathrsfs,subcaption}
\usepackage{enumitem}
\usepackage{array}
\usepackage{authblk}
\newtheorem{lemma}{Lemma}

\newtheorem{theorem}{Theorem}

\newcolumntype{M}[1]{>{\centering\arraybackslash}m{#1}}

\title{A Tight 4/3 Approximation for Capacitated Vehicle Routing in Trees}
\author{Amariah Becker \thanks{amariah\_becker@brown.edu
}}
\affil{Department of Computer Science, Brown University}

\begin{document}

\maketitle

\begin{abstract}
	Given a set of clients with demands, the {\sc Capacitated Vehicle Routing} problem is to find a set of tours that collectively cover all client demand, such that the capacity of each vehicle is not exceeded and such that the sum of the tour lengths is minimized.  In this paper, we provide a $4/3$-approximation algorithm for {\sc Capacitated Vehicle Routing} on trees, improving over the previous best-known approximation ratio of $(\sqrt{41}-1)/4$ by Asano et al.\cite{asano}, while using the same lower bound.  Asano et al. show that there exist instances whose optimal cost is $4/3$ times this lower bound.  Notably, our $4/3$ approximation ratio is therefore tight for this lower bound, achieving the best-possible performance.
\end{abstract}

\section{Intro}

Vehicle-routing problems address how a service can best be provided to meet the demand from a set of clients.  These problems arise very naturally in both commercial and public planning.  Real world vehicle-routing problems must account for the capacities of the vehicles, which limit the amount of client demand that can be met in a single trip.  We formalize this problem as finding tours in a graph:

Given a graph $G = (V,E)$, a specified \emph{depot} vertex $r$, a \emph{client set} $S$, an edge length function $l: E
\rightarrow \mathbb{Z}^{\geq 0}$, a demand function $d: S
\rightarrow \mathbb{Z}^{\geq 0}$ and $Q>0$, the {\sc Capacitated Vehicle Routing} problem
is to find a set of tours of minimum total length such that each tour includes $r$ and the tours collectively
\emph{cover} all demand at every client and such that no tour covers more than $Q$
demand.   A tour can only cover demand from clients along the tour, but it may pass by some clients without covering their demand.  

There are two common variants of this problem: \emph{splittable} and \emph{unsplittable}.  In the splittable variant, the demand of a client can be collectively covered by multiple tours, and in the unsplittable variant, the entire demand of a client must be covered by the same tour.

Both variants of this problem are NP-hard and therefore unlikely to admit a polynomial-time exact solution, but constant factor approximations can be found in polynomial time~\cite{haimovich}.  A natural question is whether better performance can be achieved for restricted graph classes.  One line of research has focused on approximating {\sc Capacitated Vehicle Routing} in trees.  Though the problem remains NP-hard in trees~\cite{labbe}, better constant-factor approximations have been found than for general metrics.  

Hamaguchi and Katoh~\cite{hamaguchi} noted a simple lower bound for the splittable-variant of {\sc Capacitated Vehicle Routing} in trees: every edge must be traversed by at least enough vehicles to accommodate all demand from the clients that use the edge on the shortest path to the depot.  They then use this lower bound (denoted $LB$) to give a 1.5 approximation~\cite{hamaguchi}.  Following this work, Asano et al.~\cite{asano} use the same lower bound to achieve a $(\sqrt{41}-1)/4$-approximation.  They also prove the following lemma:

\begin{lemma}[\cite{asano}]\label{lem:tight}
There exist instances of {\sc Capacitated Vehicle Routing} in trees whose optimal solution costs $4/3\cdot LB$.
\end{lemma}

This shows that the best possible approximation ratio using this lower bound would be a $4/3$-approximation.  Our result, stated in Theorem~\ref{thm:main}, achieves this ratio, and is therefore tight with respect to $LB$.  No further improvements over our result can be made until a better lower bound is found.

\begin{theorem}\label{thm:main} There is a polynomial-time 4/3 approximation for 
{\sc Capacitated Vehicle Routing} in trees.
\end{theorem}

\subsection{Related Work}

As {\sc Capacitated Vehicle Routing} generalizes the {\sc Traveling Salesman Problem (TSP)} (which is the special case of $Q=|V|$) it is NP-hard, and in general metrics is APX-hard~\cite{papadimitriou}.  

For general metrics, a technique called \emph{Iterated Tour Partitioning} starts with a TSP solution, partitions this tour into paths of bounded capacity, and then makes vehicle routing tours by adding paths from the depot to each endpoint of each path~\cite{haimovich}.  Iterated Tour Partitioning results in a polynomial time $(1+(1-1/Q)\alpha)$-approximation for splittable {\sc Capacitated Vehicle Routing}, where $\alpha$ is the approximation ratio of the TSP tour. A similar approach can be used for the unsplittable variant, resulting in a $(2+(1-2/Q)\alpha)$-approximation~\cite{altinkemer1987}.  Using Christofides' 1.5-approximation for TSP ~\cite{christofides}, these ratios are $(2.5-\frac{1.5}{Q})$ and $(3.5-\frac{3}{Q})$ respectively.  No significant improvements over iterated tour partitioning are known for general metrics.

Even in trees, splittable {\sc Capacitated Vehicle Routing} is NP-hard by a reduction from bin packing ~\cite{labbe}, and unsplittable {\sc Capacitated Vehicle Routing} is NP-hard to even approximate to better than a 1.5-factor~\cite{golden}.  Since depth-first search trivially solves TSP optimally in trees, iterated tour partitioning already gives a $(2-\frac{1}{Q})$-approximation for splittable demands in trees and a $(3-\frac{2}{Q})$-approximation for unsplittable demands in trees.  Labbe et al. improved this to a 2-approximation for the unsplittable-demand variant~\cite{labbe}.  For splittable {\sc Capacitated Vehicle Routing} in trees, Hamaguchi and Katoh ~\cite{hamaguchi} define a natural lower bound on the cost of the optimal solution and give a 1.5-approximation algorithm that yields a solution with cost at most 1.5 times this lower bound.  Asano, Katoh, and Kawashima~\cite{asano} improve the ratio to $(\sqrt{41}-1)/4$ using this same lower bound.  They also show that there are instances whose optimal cost is, asymptotically, $4/3$ times this lower bound value.  This implies that if a $4/3$-approximation algorithm exists that uses this lower bound it would be \emph{tight} (i.e. best possible).  We answer this open question in this paper by providing such a result.

All of the above results allow for arbitrary capacity $Q$.  Even for fixed capacity $Q\geq 3$, {\sc Capacitated Vehicle Routing} is APX-hard in general metrics~\cite{asano1997}.  For fixed capacities, {\sc Capacitated Vehicle Routing} is polynomial-time solvable in trees, but is NP-hard in other metrics.  For instances in the Euclidean plane ($\mathbb{R}^2$), polynomial-time approximation schemes (PTAS) are known for instances where $Q$ is constant ~\cite{haimovich}, $Q$ is $O(\log n/\log\log n)$~\cite{asano1997}, and $Q$ is $\Omega(n)$~\cite{asano1997}.  For higher-dimensional Euclidean spaces $\mathbb{R}^d$, a PTAS is known for when $Q$ is $O(\log^{1/d}n)$~\cite{khachay2016}.  While a \emph{quasi}-polynomial-time approximation (QPTAS) is known for arbitrary $Q$ on instances in the  Euclidean plane ($\mathbb{R}^2$)~\cite{das2010}, no PTAS is known for arbitrary $Q$ in \emph{any} non-trivial metric.

Recently, the first approximation schemes for non-Euclidean metrics were designed.  Specifically, a quasi-polynomial time approximation scheme is known for planar and bounded-genus graphs when $Q$ is fixed ~\cite{beckerquasi}, and a polynomial-time approximation scheme is known for graphs of bounded highway-dimension when $Q$ is fixed ~\cite{becker_hwy_dim}.

\subsection{Techniques}

Our work extends the techniques introduced by Asano et al. ~\cite{asano} which itself was an extension of the work of Hamaguchi and Katoh~\cite{hamaguchi}.  Specifically, Hamaguchi and Katoh describe a very natural lower bound that arises on tree instances~\cite{hamaguchi}: the number of tours that traverse each edge must be at least enough to cover all demand in the subtree below the edge (the tree is assumed to be rooted at the depot).  This introduces a minimum \emph{traffic} value on the edge.  Multiplying this value by two (each tour crosses each edge once in each direction) times the weight of the edge and summing over all edges provides a lower bound on the cost of any feasible solution.

The algorithm of Asano et al.~\cite{asano} proceeds in a sequence of rounds.  In each round, a set of tours is identified such that the ratio of the cost of these tours to the reduction to the lower bound that results from covering the demand on these tours is bounded by some constant $\alpha$.  The key is that although a given tour itself may cost more than $\alpha$ times its reduction to the lower bound, collectively the set of tours has the desired ratio.  If after a round ends, no uncovered demand remains, then the union of these sets of tours is a feasible solution with cost at most $\alpha$ times the lower bound, which is at most $\alpha$ times the optimal cost.  For Asano et al. ~\cite{asano}, $\alpha = (\sqrt{41}-1)/4$.  We use a similar approach to achieve $\alpha = 4/3$.

Asano et al. ~\cite{asano} also introduced the idea of making \emph{safe} modifications to the instance.  That is, modifying the structure of the instance in such a way that does not increase the value of the lower bound or decrease the optimum cost (although it may increase the optimum cost) and such that a feasible solution in the modified instance has a corresponding feasible solution in the original instance.  These modifications can be made safely at any point in the algorithm.  Our algorithm also makes use of safe modifications, although the ones that we define differ somewhat from those defined by Asano et al.~\cite{asano}.

These modifications allow us to reason better about how the resulting instance must be structured.  The idea is that the modifications can be made until one of a few cases arise.  Each case has a corresponding \emph{strategy} to find a set of tours with the desired cost-to-savings ratio.  

Specifically, the algorithm of Asano et al.~\cite{asano} classifies the \emph{leafmost} subtrees containing at least $2Q$ units of demand into one of a few cases.  The main obstacle in extending their  algorithm to a $4/3$-approximation is that one of the cases does not seem to have a good strategy.  On the other hand, modifying the algorithm to instead classify the \emph{leafmost} subtrees containing at least $\beta Q$ units of demand for some $\beta >2$ can greatly increase the number of cases that arise.

We overcome this obstacle by generalizing the difficult case into what we can a \emph{$p$-chain} (See Figure~\ref{fig:p_chain}).  Our key insight is that even arbitrarily large $p$-chains can be addressed efficiently in sibling pairs and at the root.  Our algorithm effectively delays addressing the difficult cases by pushing it rootward until it finds a pair or reaches the root and is thus easy to address.  To keep the number of cases small, we address easy cases as they emerge, and require that any remaining difficult case must have a specific structure that the algorithm can easily detect in subsequent rounds.

\section{Preliminaries}

We use \emph{$OPT$} to denote the cost of an optimization problem.  For a minimization problem, a polynomial-time $\alpha$-approximation algorithm is an algorithm with a runtime that is polynomial in the size of the input and returns a feasible solution with cost at most $4/3\cdot OPT$.

When the input graph $G$ is a tree it is assumed to be rooted at the depot $r$.  Let $P[u,v]$ denote the unique path from $u$ to $v$ in the tree.  Recall that the problem gives a length $l(e)\geq 0$ for each edge $e$, and let  $l(P[u,v]) = \sum_{e \in P[u,v]}l(e)$ denote the shortest path distance between $u$ and $v$.

For rooted tree $T=(V,E)$, and $v\in V$ let $T_v$ denote the subtree rooted
at $v$ and $d(T_v)$ be the total demand from vertices in $T_v$.  

The parent of a vertex $v$ is the vertex $u$ adjacent to $v$ in $P[v,r]$, and the parent of $r$ is undefined. An edge labeled $
(u,v)$ indicates that $u$ is the parent of $v$.  
The parent of an edge $(u,v)$ is the edge $(parent(u),u)$ and is undefined if $u=r$.  If $v$ has parent vertex $u$, we call $B = T_v \cup \{(u,v)\}$ a
\emph{branch} at $u$, and edge $(u,v)$ the \emph{stem} of the branch.  The parent of a branch at $u$ with stem $e$ is the branch at $parent(u)$ with stem $parent(e)$ and is undefined if $u=r$.  A vertex (resp. edge, branch) with a parent is said to be a child vertex (resp. edge, branch) of that parent.

\subsection{Lower Bound}
Since all demand must be covered and each tour can cover at most $Q$ demand, each
edge $e = (u,v)$ must be traversed by enough tours to cover $d(T_v)$ demand.  
We call this value the \emph{traffic} on the edge $e$, denoted $f(e)$.
Namely,
$f(e) = \lceil\frac{d(T_v)}{Q}\rceil$ tours.   Each such tour
traverses the edge exactly twice
(once in each direction).  We say that the lower bound $LB(e)$ of the contribution
of edge $e = (u,v)$ to the total solution cost is therefore,  $$LB(e) = 2\cdot l(e)\cdot
f(e) = 2\cdot l(e)\lceil
\frac{d(T_v)}{Q}\rceil$$ and that the lower bound $LB$ on $OPT$ is $$LB
= \sum_{e\in E}LB(e)$$
For convenience, we scale down all demand values by a factor of $Q$ and set $Q=1$.
 We also assume that the vertices with positive demand are exactly the leaves:  If
 some internal vertex $v$ has positive demand $d(v)$ we can add a vertex $v'$ with
 demand $d(v') = d(v)$ and edge $(v,v')$ of length zero and set $d(v)$ to zero.  Alternatively,
 if some leaf $v$ has zero demand, no tour in an optimal solution will visit $v$, so
 $v$ and the edge to $v$'s parent can be deleted from the graph.  Finally, we assume that no non-root vertex has degree exactly two, as no branching would occur at such a vertex, so the two incident edges can be spliced into one.

A very high level description of the algorithm is as follows: iteratively identify
sets of tours
in which the ratio of the cost of the tours to the reduction in cost to the lower bound,
$LB$, is at most $4/3$.  We call such a set of tours a $4/3$-\emph{approximate tour
set}.

We say that a $4/3$-approximate tour set \emph{removes} the demand that the tours cover.  If such a tour set removes all demand from a branch, we say that this branch has been \emph{resolved}.  After a branch is resolved, it is convenient to think of it as having been deleted from the tree and proceed with the smaller instance.

We note that any $4/3$-approximation algorithm for {\sc Capacitated Vehicle Routing} trivially generates a $4/3$-approximate tour set.  The converse, is also straightforward:

\begin{lemma}\label{lem:tours}
If, after iteratively finding and removing demand from $4/3$-approximate tour sets,
no
demand remains, then
the union of all tours sets is a $4/3$-approximation for {\sc Capacitated Vehicle Routing}.
\end{lemma}

Given this, we make one more simplifying assumption that each leaf $v$ has demand
$d(v) < 1$. Assume to the contrary that for some leaf $v$, $d(v) \geq 1$.  A tour
that goes directly from $s$ to $v$ and back and covers one unit of demand at $v$ is
in fact a 1-approximate tour (and thus also a $4/3$-approximate tour).  If ever such
a leaf exists, we can greedily take such tours until no more such leaves exist~\cite{asano}.

\subsection{Safe Operations}
We say that an operation that modifies the graph is \emph{safe} if it does not decrease
$OPT$ or change the cost of the lower bound $LB$ and it preserves feasibility.  Note that a $4/3$-approximate
tour set in the modified graph is therefore also a $4/3$-approximate tour set in the unmodified
graph.

The algorithm proceeds iteratively.  In each iteration, the algorithm performs a series
of \emph{safe} operations and takes a sequence of $4/3$-approximate tours.

We now define the operations.  Note that some of these operations also appear in~\cite{asano} with different vocabulary.

\begin{itemize}
\item {\bf Condense}:  If edge $e = (u,v)$ has traffic $f(e) = 1$ and $v$ is not a
leaf, add a vertex $v'$ and replace $e$ and $T_v$ with an edge $e'= (u,v')$ with $l
(e') = l(e) + \sum_{e''\in T_v}l(e'')$ and set $d(v') = d(T_v)$ (see Figure~\ref{fig:condense}).
 \item {\bf Unzip}: If edge $e = (u,v)$ has traffic equal to the sum of the traffic
 on child edges $(v,w_1),(v,w_2),...,(v,w_k)$, then delete $v$ and add edges $(u,w_1),
 (u,w_2),...,(u,w_k)$ with lengths $l((u,w_i))=l(e)+l(v,w_i)$ for all $i \in \{1,2,...,k\}$
 (see Figure~
\ref{fig:unzip}).  
 \item {\bf Group}:  If vertex $u$ has at least four children, including three leaf children $v_1,v_2,v_3$, such that $1.5<d(v_1)+d(v_2)+d(v_3)<2$, add vertex $u'$,
 edge $(u,u')$ of length zero, and for $i \in {1,2,3}$, replace $(u,v_i)$ with $(u',v_i)$ (see Figure~\ref{fig:group}).
\item {\bf Unite}: If vertex $u$ has leaf children $v_1$ and $v_2$ such that $d(v_1)
+ d(v_2) \leq 1$, delete $v_1$ and $v_2$ and add vertex $v_0$
with demand $d(v_1)+d(v_2)$
and edge $(u,v_0)$ with length $l((u,v_1))+l((u,v_2))$ (see Figure~\ref{fig:unite}).
\item {\bf Slide}: If edge $e_0 = (u,v)$ has child edges $ e_1 = (v,w_1)$ and $e_2
= (v,w_2)$ such that traffic $f(e_0) = f(e_1)$, then delete edge $e_2$ and add edge
$(w_1,w_2)$ of length $l(e_2)$ (see Figure~
\ref{fig:slide}).
\end{itemize}


\begin{figure}[!h]
\centering
    \begin{subfigure}[t]{0.3\textwidth}
      \includegraphics[width=\textwidth]{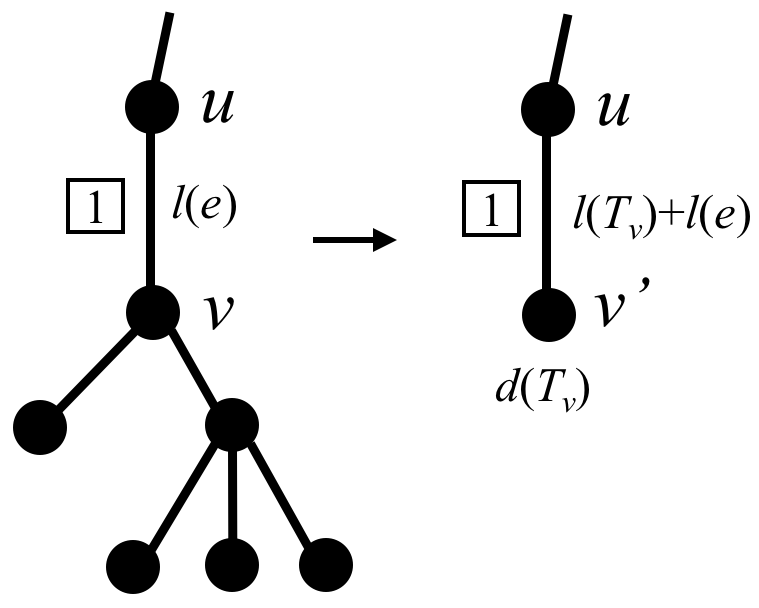}
      \caption{Condense}
      \label{fig:condense}
    \end{subfigure}
    \quad\vline\quad
    \begin{subfigure}[t]{0.35\textwidth}
      \includegraphics[width=\textwidth]{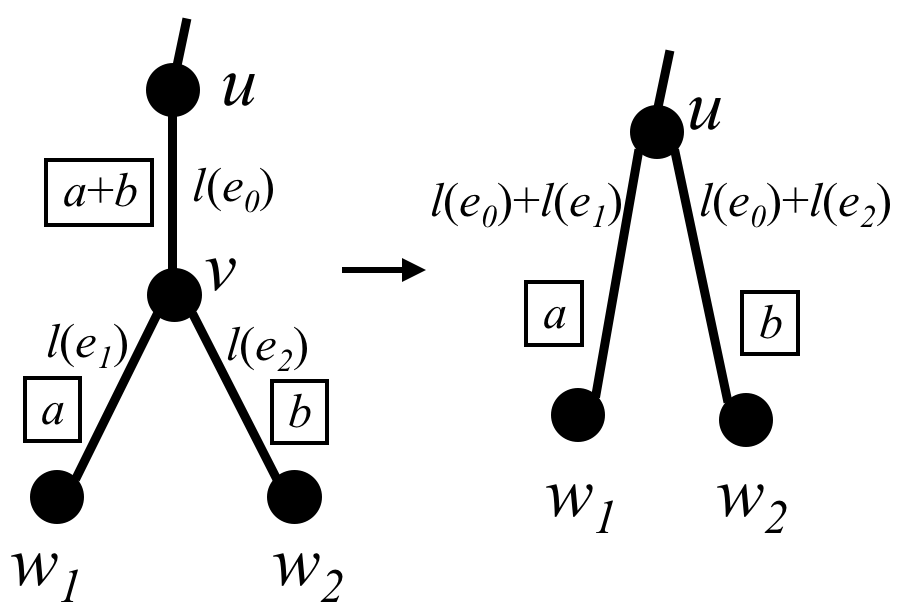}
      \caption{Unzip}
      \label{fig:unzip}
    \end{subfigure}

    \hrulefill

	\begin{subfigure}[t]{0.3\textwidth}
      \includegraphics[width=\textwidth]{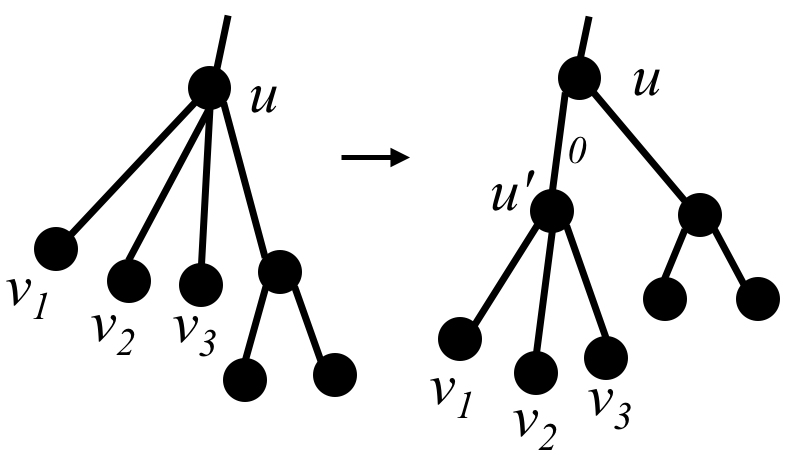}
      \caption{Group}
      \label{fig:group}
    \end{subfigure}
        \quad\vline\quad
    \begin{subfigure}[t]{0.3\textwidth}
      \includegraphics[width=\textwidth]{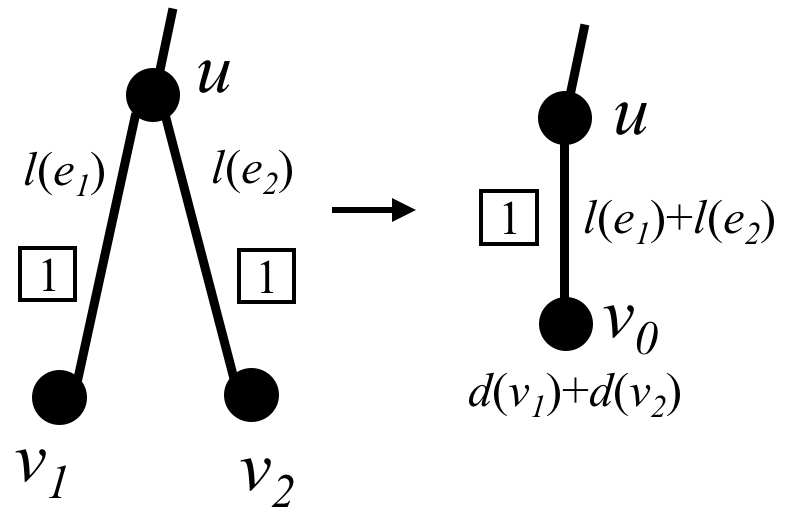}
      \caption{Unite}
      \label{fig:unite}
    \end{subfigure}
    \quad\vline\quad
    \begin{subfigure}[t]{0.25\textwidth}
      \includegraphics[width=\textwidth]{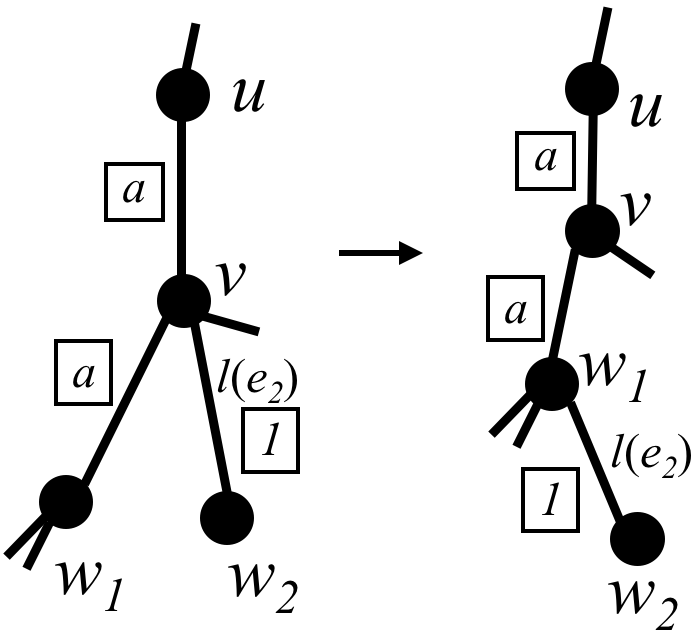}
      \caption{Slide}
      \label{fig:slide}
    \end{subfigure}
\caption{\label{fig:operations} Safe Operations.  Traffic values are shown in rectangles.}
\end{figure}

\section{Algorithm}
Exhaustively applying safe operations is called \emph{simplifying} the instance.  Note that none of the operations cancel each other, so this process terminates.

We say that a problem instance is \emph{simplified} if no more safe operations are
available, no internal vertices have demand, no non-root vertex has degree two, and for every leaf $v$, $0 < d(v) <1$.  We say that a branch is simplified if these conditions hold for the branch.

Recall that a \emph{branch} consists of a subtree along with its parent edge (\emph{stem}).
 If the branch has traffic $p$, we call it a $p$-\emph{branch}.

A simplified 2-branch with stem $e_0 = (u,v)$ such that $v$ has exactly three children $w_1,w_2,w_3$,
all of which are leafs and such that $1.5 < d(w_1)+d(w_2)+d(w_3) \leq 2$ is called a
$2$-\emph{chain}.

\begin{lemma}\label{lem:2_chain}
In a simplified problem instance, all 2-branches are 2-chains.
\end{lemma}

\begin{proof}
Consider any 2-branch in a simplified problem instance.  Clearly no edge in the branch
can have traffic greater than
 two.  If more than one child edge of the stem had traffic two, then the traffic of
 the stem itself must be greater than two.  If the stem had exactly one child edge
 with traffic two, then a slide operation would have been possible. Therefore every
 child edge of the stem has traffic one. Each of these edges must be leaf edges or
 else they could be condensed.  If there were exactly two such
 edges, then
 the stem could be unzipped.  Since no two of these edges can be united, then the
 demand of every pair sums to more than one, so there are exactly three such edges
 and their demand sums to more than 1.5.
\end{proof}

\subsection{$p$-Chains}

We now generalize the notion of a 2-chain.  For $p \geq 3$ a $p$-\emph{chain} is a
simplified $p$-branch
with stem $e_0 = (u,v)$ such that $v$ has exactly three children $w_1,w_2,w_3$, in
which $w_2$ and $w_3$ are leaves with $1<d(w_2)+d(w_3)\leq1.5$ and $(v,w_1)$ is the stem
of a $(p-1)$-chain (see Figure~\ref{fig:chain_label}).

For convenience, we define a labeling scheme for $p$-chains.  Each vertex $v_i^j$
is doubly
indexed by level $i$ and rank (child-order) $j$.  We use $e_i^j$ to denote the parent
edge of $v_i^j$, $\forall i,j$. A $2$-chain with parent $u$
has stem $e_2^0=(u,v_2^0)$ leaves $v_1^0$, $v_1^1$, and $v_1^2$ such that $l(e_1^0) \geq
l(e_1^1) \geq l
(e_1^2)$.  For $p>2$, a $p$-chain with parent $u$
has stem $e_p^0=(u,v_p^0)$, and children $v_{p-1}^0$, $v_{p-1}^1$, and $v_{p-1}^2$, such
that $e_{p-1}^0$ is the stem of a $p-1$-chain, and $v_{p-1}^1$ and $v_{p-1}^2$ are
leaves with $l(e_{p-1}^1) \geq l(e_{p-1}^2)$ (See Figure~\ref{fig:chain_label}).

\begin{figure}[!h]\label{fig:p_chain}
\centering
    \begin{subfigure}[t]{0.4\textwidth}
      \includegraphics[width=\textwidth]{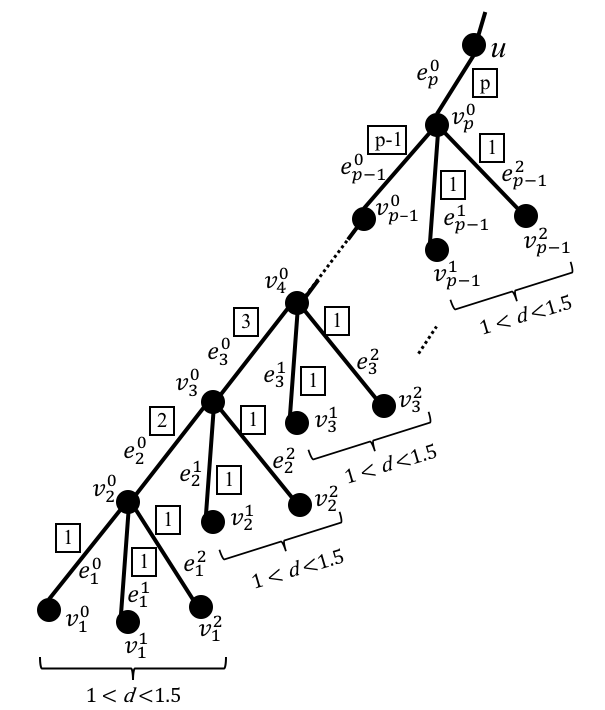}
      \caption{$p$-Chain.  Traffic values are shown in rectangles, and $d$ denotes demand.}
      \label{fig:chain_label}
    \end{subfigure}
    \quad\vline\quad
    \begin{subfigure}[t]{0.4\textwidth}
      \includegraphics[width=\textwidth]{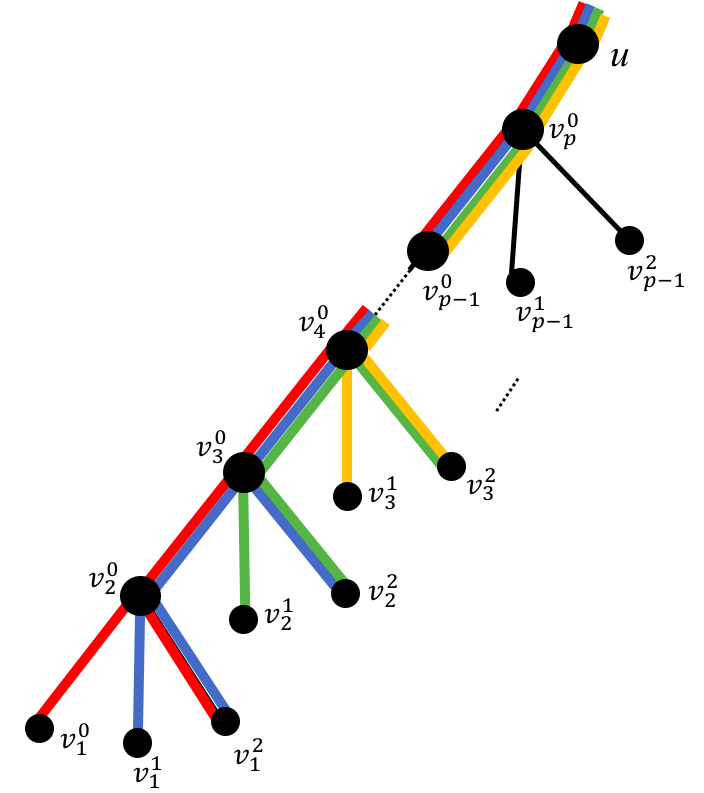}
      \caption{Cascade. The first four tours of a cascade tour set are depicted.}
      \label{fig:cascade}
    \end{subfigure}

\caption{\label{fig:p_chain}}    
\end{figure}

We further classify some  $p$-chains as \emph{long} $p$-chains:  All $2$-chains are long, and for $p \geq 3$ a \emph{long} $p$-\emph{chain} is a
$p$-chain in which $l(e_{p-1}^2) < l(P[v_{p}^0,r])$ and $e_{p-1}^0$
is the stem of a \emph{long} $(p-1)$-chain. A $p$-chain in which $l(e_{p-1}^2) \geq l(P[v_{p}^0,r])$ is called a \emph{short} $p$-chain.

Long $p$-chains are particularly convenient because they can be \emph{resolved} individually at the root and in sibling pairs for internal vertices, as described in the following lemmas (which we prove in Section~\ref{sec:resolving_pchains}).

\begin{lemma}\label{lem:root_pchain}
Long $p$-chains can be resolved at the root.
\end{lemma}

\begin{lemma}\label{lem:sibling_pchain}
A long $p$-chain and long $p'$-chain can be resolved together if they are sibling
branches. 
\end{lemma}

As in~\cite{asano}, our algorithm proceeds in a series of iterations.  Each iteration
performs a set of safe operations and identifies a $4/3$-approximate
tour set.

\begin{lemma}\label{lem:iteration}
Iteration $i$ runs in polynomial time and either finds a nonempty $4/3$-approximate
tour set or finds that every branch at the root is either a long $p$-chains or 1-branch.
\end{lemma}
\begin{proof}
See Section~\ref{sec:iteration}.
\end{proof}

These iterations continue until every branch at the root is either a long $p$-chains or 1-branch. The algorithm then solves the remaining instance, using
the result from Lemma~\ref{lem:root}.

\begin{lemma}\label{lem:root}
There is a polynomial time $4/3$-approximation algorithm for instances of {\sc Capacitated Vehicle Routing} on
trees in which every child branch of the root is either a long $p$-chain or a 1-branch.
\end{lemma}
\begin{proof}
The cost of a tour that traverses a 1-branch at the root is equivalent to the reduction to lower bound $LB$ that results from removing the demand of the branch, so such a tour is a 1-approximate tour (and thus also a $4/3$-approximate tour). The lemma result then follows from Lemma~\ref{lem:tours} and Lemma~\ref{lem:root_pchain}
\end{proof}

Putting these steps together, gives our overall 4/3-approximation result described
in Theorem~\ref{thm:main}.

\setcounter{theorem}{0}
\begin{theorem}\label{thm:main} There is a polynomial-time 4/3 approximation for 
{\sc Capacitated Vehicle
Routing} in trees.
\end{theorem}
\begin{proof}
Let $m$ denote the total amount of demand in the graph. Since each iteration removes demand, there are at most $m$ iterations, each of which runs in polynomial time.  By Lemma~\ref{lem:iteration},
after iteration $m$, every branch off the root must be a long $p$-chain or
a 1-branch.  By Lemma~\ref{lem:root} there is a polynomial-time $4/3$-approximation for these branches.
Combining this approximation with the collection of tours identified during the iterations
results in an overall $4/3$ approximation, by \ref{lem:tours}.
\end{proof}

All that remains is to prove the above lemmas.  In Section~\ref{sec:iteration} we describe the iteration subroutine and prove Lemma~\ref{lem:iteration}.  In Section~\ref{sec:resolving_pchains} we prove Lemmas~\ref{lem:root_pchain} and \ref{lem:sibling_pchain}.

\section{Description of Iteration $i$}\label{sec:iteration}

We say that a $p$-branch $B$ is \emph{settled} if it is either a 1-branch or a long $p$-chain.  Otherwise we say that $B$ is \emph{unsettled}.  We say that a $p$-branch $B$ is \emph{minimally unsettled} if it is unsettled and all of its child branches are settled.

Each iteration consists of the following subroutine:
\begin{enumerate}
\setlength{\itemsep}{0pt}
\item Simplify the instance (i.e. exhaustively apply safe operations)
\item If all branches at the root are settled, terminate iteration.
\item Otherwise, find a minimally unsettled branch $B$.
\begin{enumerate}[label = case (\roman*)]
\setlength{\itemsep}{0pt}
\item If $B$ has at least two child branches that are long $p$-chains, apply Lemma~\ref{lem:sibling_pchain}.
\item Otherwise, if $B$ has at least three child branches that are $1$-branches, apply Lemma~\ref{lem:trident}
\item Otherwise, $B$ is a (short) $p$-chain.  Apply Lemma~\ref{lem:short_pchain}.
\end{enumerate}
\end{enumerate}

\begin{lemma}\label{lem:trident} A simplified, minimally unsettled branch $B$ with at least three child branches that are 1-branches admits a $4/3$-approximate tour set.
\end{lemma}
\begin{proof}
Let $(u,v_0)$ be the stem of $B$, and let $(v_0,v_1)$, $(v_0,v_2)$, and $(v_0,v_3)$ be three (child) 1-branches.  Since the branch is simplified, $v_1$, $v_2$, and $v_3$ are leaves.  Since the unite operation is unavailable, $d(v_1)+d(v_2)+d(v_3)>1.5$, and since $B$ is unsettled, it is not a 2-chain.  Furthermore, the group operation is unavailable, so $2 < d(v_1)+d(v_2)+d(v_3) < 3$.  Let $a = l(P[v_0,r])$, $w_1 = l(v_0,v_1)$, $w_2 = l(v_0,v_2)$, and $w_3 = l(v_0,v_3)$.  Without loss of generality, assume $w_1 \leq w_2 \leq w_3$.

If $a \leq w_1 + w_2 + w_3$, then for $i \in \{1,2,3\}$, let $t_i$ be the tour that travels from the depot to $v_i$, covers all demand at $v_i$, and then returns to the depot.  The length of $t_i$ is $2(a+w_i)$.  The total cost of the tour set $\{t_1,t_2,t_3\}$ is $2(3a + w_1 + w_2 + w_3)$.  This tour set covers all demand of the leaves and also reduces demand along $P[v_0,r]$ by two, since $2 < d(v_1)+d(v_2)+d(v_3) < 3$, so the reduction to $LB$ is $2(2a + w_1 + w_2 + w_3)$.  The ratio of the cost of the tour set to the reduction to the lower bound is therefore,
$$\frac{2(3a + w_1 + w_2 + w_3)}{2(2a + w_1 + w_2 + w_3)} = 1+ \frac{a}{2a + w_1 + w_2 + w_3} \leq \frac{4}{3}$$
where the final equality comes from $a \leq w_1 + w_2 + w_3$.

Otherwise, $a > w_1 + w_2 + w_3$.  Let $t$ be the tour that travels from the depot to $v_1$, covers all demand at $v_1$, travels to $v_3$, covers as much demand as possible at $v_3$, and then returns to the depot.  The cost of $t$ is $2(a+w_1+w_3)$.  Note that since the unite operation is unavailable, then $d(v_1)+d(v_3) > 1$, so the vehicle is full, and some demand remains at $v_3$.  Since the vehicle is full, then $t$ reduces demand along $P[v_0,r]$ by one, so the reduction to $LB$ is $2(w_1 + a)$.  The ratio of the cost of the tour set $\{t\}$ to the reduction to the lower bound is therefore, 

$$\frac{2(a+w_1+w_3)}{2(w_1 + a)} = 1 + \frac{w_3}{w_1+a} \leq \frac{4}{3}$$

where the final inequality comes from $a > w_1 + w_2 + w_3 \geq w_1 \geq w_3$.

\end{proof}

\begin{lemma}\label{lem:short_pchain} A short $p$-chain admits a $4/3$-approximate tour set.
\end{lemma}
\begin{proof}
Let $a = l(P[v_{p}^0,r])$, $b = l(e_{p-1}^2)$, and $c=l(e_{p-1}^1)$.  By construction, $c \geq b$, and since the $p$-chain is short, $b \geq a$.  Let $t_1$ be the tour that goes from the depot to $v_{p-1}^1$, covers all demand at this leaf, and then returns to the depot.  Similarly, let $t_2$ be the tour that goes from the depot to $v_{p-1}^2$, covers all demand at this leaf, and then returns to the depot. Tour $t_1$ has cost $2(a+c)$ and $t_2$ has cost $2(a+b)$.  Tour set $\{t_1,t_2\}$ has total cost $2(2a + b + c)$.  This tour set covers all demand at these leaves, which sum to greater than one by definition of $p$-chain, so these tours reduce the demand along $P[v_{p}^0,r]$ by one.  Therefore the reduction to $LB$ from this tour set is $2(a+b+c)$.  Therefore the ratio of the cost of the tour set to the reduction to $LB$ is,
$$\frac{2(2a + b + c)}{2(a + b + c)} = 1 + \frac{a}{a+b+c} \leq \frac{4}{3}$$

The final inequality comes from $a\leq b\leq c$.

A $p$-chain in which $l(e_{p-1}^2) \geq l(P[v_{p-1}^2,r])$ is called a \emph{short} $p$-chain.

\end{proof}

Finally, we prove Lemma~\ref{lem:iteration}, which we restate here for convenience.

\setcounter{lemma}{5}
\begin{lemma}
Iteration $i$ runs in polynomial time and either finds a nonempty $4/3$-approximate
tour set or finds that every branch at the root is either a long $p$-chains or 1-branch.
\end{lemma}
\begin{proof}
To prove this lemma, we first must show that any minimally unsettled branch $B$ in a simplified instance must be in at least one of the three identified cases.  Let $B$ be such a branch. $B$ must have child branches, since it is unsettled. Since it is minimally unsettled, every child branch is settled (i.e. either a 1-branch or a long $p$-chain). 

If $B$ has at least two long $p$-chains as child branches, then case $i$ holds.

If $B$ has no long $p$-chains as child branches, then all child branches are 1-branches.  By Lemma~\ref{lem:2_chain}, in a simplified instance all 2-branches are 2-chains.  Since all 2-chains are long, and thus settled, $B$ must be a $j$-branch for some $j \geq 3$.  Since the unzip operation is unavailable, then there are at least four 1-branches as child branches of $B$, so case $ii$ holds.

If $B$ has exactly one long $p$-chain as a child branch and case $ii$ doesn't hold, then $B$ must be a short $p'$-chain.  If there were no 1-branches as child branches, then the degree-two vertex could be spliced.  If there were exactly one 1-branch then either $p=p'$ and a slide operation would be available, or $p = p'-1$ and an unzip operation would be available. Since case $ii$ doesn't hold, there are exactly two 1-branches as child branches.  Since an unzip operation is unavailable, $p = p'-1$, and since a unite operation is unavailable, the demand of the two 1-branches (which are both leaves by the condense operation) sum to greater than one. Since the instance is simplified, then $B$ is a $p'$-chain.  But since $B$ is unsettled, it must be a short $p'$-chain.  Therefore case $iii$ holds.

Since $B$ is covered by some case, then the corresponding lemma (Lemma \ref{lem:sibling_pchain}, \ref{lem:trident}, or \ref{lem:short_pchain}) guarantees a nonempty $4/3$-approximate tour set.

If no unsettled branch can be found, then all branches at the root must be settled.

Finally, to see that the iteration runs in polynomial time, note that each operation can be performed in constant time.  Condense, unite, unzip, (and splicing) all make the graph smaller, so they are exhausted in linear time.  Subsequently, if condense and unite are unavailable, then each slide results in a degree two vertex remaining to be spliced, also making the graph smaller.  After all other operations are exhausted, each leaf can participate in at most one group operation in a given iteration.  Each of the three cases can likewise be handled in linear time.
\end{proof}

\section{Resolving $p$-Chains}\label{sec:resolving_pchains}

\subsection{Cascades}\label{sec:cascades}

In order to prove Lemmas~\ref{lem:root_pchain} and \ref{lem:sibling_pchain} we describe a specific group of tours that collectively covers the demand of a long $p$-chain.

The labeling on $p$-chains gives a natural bottom-up ordering on the leaves: $v_1^0,v_1^1,v_1^2,v_2^1,v_2^2,$ $...,v_{p-1}^1,v_{p-1}^2$, where the order is determined first by level and then by rank.
For a long $p$-chain, we define a \emph{cascade} to be the sequence of $p$
tours in which each tour considers the leaves in this bottom-up order and:
\begin{enumerate}
\item Visits and covers all demand from the first leaf with remaining demand.
\item While the vehicle has spare capacity and there is unmet demand in the branch,
determines the lowest level $i$ with leaves
with unmet demand and covers as much demand as possible from $v_i^2$.
\item Returns to the depot.
\end{enumerate}

Let the resulting tours be $t_1,...,t_p$.  Tour $t_1$ first covers all demand from
$v_1^0$ and then covers $1-d(v_1^0)$ demand from $v_1^2$.  Since by definition, a $p$-chain is simplified, the unite operation
is unavailable, so $d(v_1^0)+d(v_1^2)
> 1$, implying the vehicle is full.  The second tour $t_2$ covers all demand from $v_1^1$,
covers all remaining demand at $v_1^2$ (since $d(v_1^0)+d(v_1^1)+d(v_1^2)
< 2$) and covers some demand at $v_2^2$.  Since the slide operation is unavailable,
$d(v_1^0)+d(v_1^1)+d(v_1^2) +d(v_2^2)
> 3$, so the vehicle is full.  Note that both vehicles have been
full, and after two tours no demand
remains at level $1$.  This pattern continues. Tour $t_i$ for $3 \leq i < p$ covers
all demand at $v_{i-1}^1$, all remaining demand at $v_{i-1}^2$, and some demand at
$v_i^2$.  Inductively, since slide operation is unavailable and all previous vehicles have
been full, this vehicle must also be full.  After $t_i$, no demand remains below level
$i$.  Since all vehicles have been full, there is less than one unit of demand remaining
for $t_p$ (and no demand remains below level $p-1$), so $t_p$ covers all demand from
$v_{p-1}^1$ and all remaining demand from $v_{p-1}^2$. (See Figure~\ref{fig:cascade}).

This cascading pattern results in $i$ tours traversing $e_i^0$,
one tour traversing $e_i^1$, two tours traversing $e_i^2$ for all
$i$, and $p$ tours traversing all edges in the path from the depot $r$ to $v_p^0$.
Note in particular that these values match the lower bounds given by the edge
traffic, for $e_i^0$ and $e_i^1$.  The excess cost, therefore, comes from doubling
the traffic lower bound on the $e_i^2$ edges as well as extra traffic along the 
path to the depot.

\subsection{Proofs of Lemmas~\ref{lem:root_pchain} and \ref{lem:sibling_pchain}}\label{sec:proofs}

We restate the lemmas here for convenience.  Recall that a branch is \emph{resolved} if all of its demand is covered by a $4/3$-approximate tour group. That is, a set of tours whose total cost is at most $4/3$ times the reduction to the lower bound $LB$ that results from removing the demand of the branch.

\setcounter{lemma}{3}
\begin{lemma}
Long $p$-chains can be resolved at the root.
\end{lemma}
\begin{proof}

Let $B$ be a long $p$-chain at the root (depot).  Consider the cascade on $B$.
 As described in Section~\ref{sec:cascades}, the cost of the $p$ tours in the cascade is 
 $$2\big(p\cdot l(e_p^0)+\sum_{i=1}^{p-1} i\cdot l(e_i^0) + 2l(e_i^2) + l(e_i^1)\big)$$
 since the
 cost of the path to the depot is zero.  Additionally the cascade covers all demand
 in $B$, so the reduction to $LB$ after covering $B$ is: $$\sum_{e\in B}LB(e)
 = \sum_{e=(u,v)\in B}2\cdot l(e)\cdot\lceil d(T_v)\rceil = 2\big(p\cdot l(e_p^0)+\sum_{i=1}^{p-1} i\cdot l(e_i^0) + l(e_i^2) + l(e_i^1)\big)$$

 Taking the ratio of cost to reduction in lower bound gives our result:

 $$\frac{2\big(p\cdot l(e_p^0)+\sum_{i=1}^{p-1} i\cdot l(e_i^0) + 2l(e_i^2) + l(e_i^1)\big)}{2\big(p\cdot l(e_p^0)+\sum_
 {i=1}^{p-1} i\cdot l(e_i^0) + l(e_i^2) + l(e_i^1)\big)}=1+\frac{\sum_{i=1}^{p-1}  l(e_i^2) }{p\cdot l(e_p^0)+\sum_
 {i=1}^{p-1} i\cdot l(e_i^0) + l(e_i^2) + l(e_i^1)}$$

 $$\leq 1+\frac{\big(l(e_1^2)\big) + \big(\sum_{i=2}^{p-1} l(e_i^2)\big) }{\big(l(e_1^0) +l
 (e_1^1)+ l(e_1^2)\big) + \big(\sum_
 {i=2}^{p-1} 3\cdot l(e_{i+1}^0) + l(e_i^2) + l(e_i^1)\big)}$$

 $$\leq 1+\frac{\big(l(e_1^2)\big) + \big(\sum_{i=2}^{p-1} l(e_i^2)\big) }{\big(3\cdot
 l(e_1^2)\big) + \big(\sum_
 {i=2}^{p-1} 3\cdot l(e_{i}^2)\big)} \leq \frac{4}{3}$$

\end{proof}

\begin{lemma}
A long $p$-chain and long $p'$-chain can be resolved together if they are sibling
branches. 
\end{lemma}
\begin{proof}

Let $u$ be a vertex with child branches $B$ a long $p$-chain and $B'$ a long $p'$-chain.

Consider the cascades on $B$ and $B'$. The cost of these $p+p'$ tours is 
$$2\big[\big(p\cdot l(e_p^0)+\sum_{i=1}^{p-1} i\cdot l(e_i^0) + 2l(e_i^2) + l(e_i^1)\big)
+ \big(p'\cdot l({e'}_{p'}^0)+\sum_{i=1}^{p'-1} i\cdot l({e'}_i^0) + 2l({e'}_i^2)
+ l({e'}_i^1)\big) + (p+p')l(P[u,r])\big]$$

Additionally the cascade covers all demand
in $B$ and $B'$, so the reduction to $LB$ is $\sum_{e\in B \cup B'}LB
(e)
= \sum_{e\in B \cup B'}2\cdot l(e)\cdot f(e) $, which is 
$$2\big[\big(p\cdot l(e_p^0)+\sum_{i=1}^{p-1} i\cdot l(e_i^0) + l(e_i^2) + l(e_i^1)\big)
+ \big(p'\cdot l({e'}_{p'}^0)+\sum_{i=1}^{p'-1} i\cdot l({e'}_i^0) + l({e'}_i^2)
+ l({e'}_i^1)\big)
+ (p+p'-1)l(P[u,r])\big]$$
Note that the $p+p'-1$ factor in the reduction to the edges along $P[u,r]$ arises because by definition the total demand in a $p$-chain is between $p-0.5$ and $p$, so covering all demand in $B$ and $B'$ reduces the demand in $T_u$ by at least $p+p'-1$.  Rearranging the terms, this reduction to the lower bound is greater than:
$$2\big[\big(\sum_{i=0}^{2}l(e_1^i)+\sum_{i=2}^{p-1} l(P[v_{i+1}^0,r]) + l(e_i^2) + l(e_i^1)\big)
$$ $$+ \big(\sum_{i=0}^{2}l({e'}_1^i)+\sum_{i=2}^{p'-1} l(P[{v'}_{i+1}^0,r]) + l({e'}_i^2)
+ l({e'}_i^1)\big)
+ 3l(P[u,r])\big]$$
$$= 2(X+ Y + X' + Y' +Z)$$

Where $X = \sum_{i=0}^{2}l(e_1^i)$, $Y = \sum_{i=2}^{p-1} l(P[v_{i+1}^0,r]) + l(e_i^2) + l(e_i^1)$, $X' = \sum_{i=0}^{2}l({e'}_1^i)$, $Y' = \sum_{i=2}^{p'-1} l(P[{v'}_{i+1}^0,r]) + l({e'}_i^2)
+ l({e'}_i^1)$, and $Z= 3l(P[u,r])$.

Taking the ratio of cost to reduction is therefore at most,
$$1+\frac{\sum_{i=1}^{p-1} l(e_i^2) + \sum_{i=1}^{p'-1} l({e'}_i^2)
+ l(P[u,r])}{X+ Y + X' + Y' +Z} \leq \frac{4}{3}$$

Where the final inequality comes from noting that, by construction, $l(e_1^2)\leq X/3$ and  $l({e'}_1^2)\leq X'/3$, and by definition of a long $p$-chain, for all $2\leq i < p$, $l({e}_i^2) \leq \min\{l({e}_i^1),l({e}_i^2),l(P[{v}_{i+1}^0,r])\}$ , so $\sum_{i=2}^{p-1} l(e_i^2) \leq Y/3$ and $\sum_{i=2}^{p-1} l({e'}_i^2) \leq Y'/3$.

\end{proof}

\bibliography{main}

\begin{thebibliography}{10}

\bibitem{altinkemer1987}
Kemal Altinkemer and Bezalel Gavish.
\newblock Heuristics for unequal weight delivery problems with a fixed error
  guarantee.
\newblock {\em Operations Research Letters}, 6(4):149--158, 1987.

\bibitem{asano}
Tetsuo Asano, Naoki Katoh, and Kazuhiro Kawashima.
\newblock A new approximation algorithm for the capacitated vehicle routing
  problem on a tree.
\newblock {\em Journal of Combinatorial Optimization}, 5(2):213--231, 2001.

\bibitem{asano1997}
Tetsuo Asano, Naoki Katoh, Hisao Tamaki, and Takeshi Tokuyama.
\newblock Covering points in the plane by $k$-tours: towards a polynomial time
  approximation scheme for general $k$.
\newblock In {\em Proceedings of the twenty-ninth annual ACM symposium on
  Theory of computing}, pages 275--283. ACM, 1997.

\bibitem{becker_hwy_dim}
Amariah Becker, Philip~N Klein, and David Saulpic.
\newblock Polynomial-time approximation schemes for $k$-center and
  bounded-capacity vehicle routing in metrics with bounded highway dimension.
\newblock {\em arXiv preprint arXiv:1707.08270}, 2017.

\bibitem{beckerquasi}
Amariah Becker, Philip~N Klein, and David Saulpic.
\newblock A quasi-polynomial-time approximation scheme for vehicle routing on
  planar and bounded-genus graphs.
\newblock In {\em LIPIcs-Leibniz International Proceedings in Informatics},
  volume~87. Schloss Dagstuhl-Leibniz-Zentrum fuer Informatik, 2017.

\bibitem{christofides}
Nicos Christofides.
\newblock Worst-case analysis of a new heuristic for the travelling salesman
  problem.
\newblock Technical report, Carnegie-Mellon Univ Pittsburgh Pa Management
  Sciences Research Group, 1976.

\bibitem{das2010}
Aparna Das and Claire Mathieu.
\newblock A quasi-polynomial time approximation scheme for {Euclidean}
  capacitated vehicle routing.
\newblock In {\em Proceedings of the twenty-first annual ACM-SIAM symposium on
  Discrete Algorithms}, pages 390--403. SIAM, 2010.

\bibitem{golden}
Bruce~L Golden and Richard~T Wong.
\newblock Capacitated arc routing problems.
\newblock {\em Networks}, 11(3):305--315, 1981.

\bibitem{haimovich}
Mark Haimovich and AHG Rinnooy~Kan.
\newblock Bounds and heuristics for capacitated routing problems.
\newblock {\em Mathematics of operations Research}, 10(4):527--542, 1985.

\bibitem{hamaguchi}
Shin-ya Hamaguchi and Naoki Katoh.
\newblock A capacitated vehicle routing problem on a tree.
\newblock In {\em International Symposium on Algorithms and Computation}, pages
  399--407. Springer, 1998.

\bibitem{khachay2016}
Michael Khachay and Roman Dubinin.
\newblock {PTAS} for the {Euclidean} capacitated vehicle routing problem in
  $\mathbb{R}^d $.
\newblock In {\em International Conference on Discrete Optimization and
  Operations Research}, pages 193--205. Springer, 2016.

\bibitem{labbe}
Martine Labb{\'e}, Gilbert Laporte, and H{\'e}lene Mercure.
\newblock Capacitated vehicle routing on trees.
\newblock {\em Operations Research}, 39(4):616--622, 1991.

\bibitem{papadimitriou}
Christos~H Papadimitriou and Mihalis Yannakakis.
\newblock The traveling salesman problem with distances one and two.
\newblock {\em Mathematics of Operations Research}, 18(1):1--11, 1993.

\end{thebibliography}

\end{document}